\documentclass[conference,a4paper]{IEEEtran}

\usepackage[cmex10]{amsmath}
\usepackage{amsthm}
\usepackage{mathrsfs}
\usepackage{cite}
\usepackage{algorithm}
\usepackage{algorithmic}

\usepackage{amsfonts}
\usepackage{amssymb}
\usepackage{amsmath}
\usepackage{bm}
\newtheorem{mydef}{Definition}
\newtheorem{theorem}{Theorem}
\newtheorem{lemma}{Lemma}

\begin{document}

\sloppy

\title{Random Linear Network Codes for Secrecy over Wireless Broadcast Channels}

\author{
  \IEEEauthorblockN{Shahriar Etemadi Tajbakhsh}
  \IEEEauthorblockA{Research School of Engineering\\
    The Australian National University\\
    Canberra, Australia\\
    Email: shahriar.etemadi-tajbakhsh@anu.edu.au}
  \and
  \IEEEauthorblockN{Parastoo Sadeghi}
  \IEEEauthorblockA{Research School of Engineering\\
    The Australian National University\\
    Canberra, Australia\\
    Email: parastoo.sadeghi@anu.edu.au}

}



\maketitle

\begin{abstract}
   We consider a set of $n$ messages and a group of $k$ clients. Each
client is privileged for receiving an arbitrary subset of the
messages over a broadcast erasure channel, which generalizes
scenario of a previous work. We propose a method for secretly
delivering each message to its privileged recipients in a way that
each receiver can decode its own messages but not the others'. Our
method is based on combining the messages using linear network
coding and hiding the decoding coefficients from the unprivileged
clients. We provide an information theoretic proof for the secrecy
of the proposed method. In particular we show that an unprivileged
client cannot obtain any meaningful information even if it holds the
entire set of coded data packets transmitted over the channel.
Moreover, in our method, the decoding complexity is desirably low at
the receiver side.

\end{abstract}



\section{Introduction}
Wireless medium is potentially vulnerable to different types of
security attacks as the information is broadcast in the air and
might be easily accessed or manipulated by an unprivileged party.
Specifically, secrecy of transmission sessions is a major concern
which implies the necessity of cryptographic methods against
eavesdroppers. On the other hand, recent advances in cooperative
networking schemes unveils the advantage of cooperation among
wireless devices as a consequence of diversity in wireless channels.
Despite the benefits of cooperation, it might increase the security
risks in presence of dishonest participants. As a major class of
such cooperative settings, network coding techniques require that
the wireless devices to be enabled to listen to transmission
sessions which are not necessarily intended for them and buffer what
they hear on the channel as side information. Because of diversity
in packet reception at different users, the sender might benefit
from network coding techniques to merge multiple transmission
sessions into one
\cite{COPE,BroadcastNguyen,MultipleUnicastHARQ,ARQMedard,YunnanWireless},
etc. As the users are supposed to buffer and process some messages
for which they are not their target recipients, some mechanisms
should be designed to protect the secrecy of those messages during
the cooperation.

In this paper, we consider $k$ wireless users which are connected to
a base station and share a broadcast erasure channel. Each client is
interested in receiving an arbitrary subset of $n$ messages. The
clients are enabled to listen to all the transmissions over the
channel and save what they receive in their buffers. As each client
might have missed some parts of information which it needs to decode
its own messages, either the base station should retransmit the
missing parts or the clients should cooperate with each other to
obtain the missing parts. We propose a method to maintain the
secrecy of individual messages against any unprivileged party
(either those clients who are not the target recipient of that
message or any external eavesdropper). The essence of the proposed
method is to combine all the messages together at the base station
regardless of their target recipients using a special form of random
linear network coding and broadcast the resulting packets to all the
clients; Each client privately receives a set of decoding
coefficients which enables it to decode its own messages but not the
others'. In other words, the main idea is to protect the decoding
coefficients against unprivileged parties.

  A brief summary of the contribution of this paper is as follows. We
propose a method to maintain the secrecy of transmissions over a
wireless broadcast channel by coding the messages using a special
form of random linear network coding. This paper extends the
scenario discussed in \cite{ITW2012} to a general scenario that each
client is interested in an arbitrary subset of messages. Also we
prove that using our proposed method, the eavesdropper or any
unprivileged client can not obtain any \emph{meaningful} information
about the messages. Moreover, our proof implies that the field size
of operations can be kept small which substantially reduces the
computational complexity especially at the receiver side which is a
crucial improvement over \cite{ITW2012}.

The rest of this paper is organized as follows. In section
\ref{sec:Related}, the position of this paper within the literature
is highlighted. In section \ref{sec:System} the proposed system is
introduced and some specifications and advantages of our proposed
method is discussed. Section \ref{sec:Example} provides a clarifying
example of the entire system. Finally, in section \ref{sec:Proof},
secrecy of the proposed method is proven.

\section{Related Work} \label{sec:Related}
This paper is an extension of the work in \cite{ITW2012}, where only
eavesdroppers with bounded computational power were considered to
wiretap a shared broadcast channel. In the current paper, we provide
an information theoretic proof for security of the proposed method
and we show that an eavesdropper would not be able to obtain any
meaningful information about the protected messages. Moreover, in
\cite{ITW2012} each client is only interested in a distinct message
while the current paper generalizes the proposed method to the
scenario that each of the $k$ clients is interested in receiving an
arbitrary subset of $n$ messages over a shared broadcast erasure
channel. Also, unlike \cite{ITW2012}, in this paper we operate over
a field size of $2$ as it is shown that the field size is not needed
to be large.

We take the security advantages of random linear network codes (RLNC
\cite{RandomNetCod}) in this paper. In \cite{SecureNetCod} the
achievable rate region of network coding with perfect secrecy in
presence of an eavesdropper (which is able wiretap some of the links
in a multicast scenario) is characterized. Perfect secrecy is a
strict constraint which is hardly satisfied and dramatically
degrades the throughput of a multicast network. This condition is
relaxed in \cite{WeaklySecure} to a weaker security condition which
is still satisfying in a practical sense. Weakly security guarantees
that the eavesdropper can not obtain any \emph{meaningful}
information about the messages. In this paper, we take this
definition of security to prove the secrecy of our method. The
capacity region with perfect secrecy constraint over a broadcast
channel is characterized in \cite{SecrecyLaszlo} and a method is
proposed to achieve the identified capacity region. Since we have
relaxed the condition of perfect secrecy, higher data rates can be
achieved using our proposed method.

On the other hand our proposed method is based on protecting the
decoding coefficients of linear combinations of messages from
eavesdroppers. Similar concept has been considered by
\cite{CryptAnalysisNetCod,MultiResolution} where
\cite{CryptAnalysisNetCod} provides information theoretic bounds and
theorems to guarantee that the mutual information between the
transmitted information over the links and respectively, the coding
coefficients or the original messages is small and zero under some
special conditions. In \cite{MultiResolution} a coding scheme is
proposed (based on protection of coding coefficients) for
multi-resolution video streaming where each client receives a number
of layers of the video in a successive refinement fashion according
to its subscription level. Therefore, the rest of layers should not
be revealed to this specific client. Our proposed coding algorithm
provides a substantial freedom to generalize the code protection
based methods to a scenario that each client is interested in an
arbitrary subset of messages which distinguishes our work from
\cite{CryptAnalysisNetCod,MultiResolution} (In particular, we
protect the \emph{decoding coefficients} rather than encoding
coefficients). Moreover we provide a method to update the decoding
coefficients periodically which improves secrecy.

Another important feature of our proposed scheme is its low
complexity which is a major issue due to matrix inversion operations
required to decode network coded based schemes. This problem has
been addressed in different studies such as \cite{ShojaniaParallel}
in terms of
 computational power and energy consumption limitations of
wireless devices. In our method, this burden has been shifted to the
base station which often has considerably larger computational
resources and theoretically unlimited energy than the small receiver
devices. In our scheme, the devices are only supposed to generate
linear combinations of the packets they have received over a small
size finite field.

\section{System and Model} \label{sec:System}
We consider a set of $n$ messages $X=\{x_1,\dots,x_n\}$ and a set of
$k$ clients $C=\{c_1,\dots, c_k\}$. Each client $c_i$ is interested
in receiving an arbitrary subset of messages $\chi_i\subseteq X$
from a common base station.
%
 Each message $x_i$ is composed of $T$ elements
each drawn from a finite field $\mathbb{F}_q$ of size $q$ and is
denoted by $x_i^{(t)}$. For the ease of our analysis and also to
reduce the decoding complexity at the receiver side, we assume that
all the operations are done over a finite field size of $2$, i.e.
$\mathbb{F}_2$. We consider $T$ rounds of transmission, where at
each round $1\leq t \leq T$, the set of elements
$X^{(t)}=\{x_1^{(t)},\dots, x_n^{(t)}\}$ should be delivered to the
clients at the end of round. The set of clients who are the
privileged recipients of message $x_j$ are denoted by $R_j=\{c_u:
x_j\in \chi_u\}$.

Each round of transmission incorporates three phases (1) The set of
elements $X^{(t)}$ are encoded as it will be described later and the
set of encoded elements (denoted by
$P^{(t)}=\{P_1^{(t)},\dots,P_n^{(t)}\}$) are transmitted over a
shared broadcast channel to all the clients. Each client $c_i$ might
receive each element $P_i^{(t)}$ with a probability $1-p_i$. (2) The
missing packets by the clients at each round should be retransmitted
by the either the base station or by the clients if the clients have
received the set of encoded elements collectively. (3) The base
station provides a set of decoding coefficients privately to each
client where each client is enabled to decode its own set of
elements but not the other ones', therefore the secrecy of
individual messages are maintained. In the following, the three
mentioned phases are described. In section \ref{sec:Example}, a
comprehensive example is provided to illustrate the entire process.
\begin{itemize}
\item \emph{Broadcast Phase:} At each round $t$, the base station
generates the set of encoded elements by solving the system of
equations $\mathbf{X}^{(t)}=\mathbf{A}^{(t)}\mathbf{P}^{(t)}$, where
$\mathbf{A}^{(t)}=[\alpha_{ij}]_{n\times n}$ is a matrix of randomly
chosen elements $\alpha_{ij}$ from the finite field $\mathbb{F}_q$,
$\mathbf{P}^{(t)}=[P_i^{(t)}]_{n\times 1}$ is the vector of encoded
elements and $\mathbf{X}^{(t)}=[x_i^{(t)}]_{n\times 1}$ is the
vector of message elements at round $t$. The set of encoded elements
$P^{(t)}$ is broadcast to all the clients by the base station. To
recover a message $x_i^{(t)}$, a corresponding client $c_j$ for whom
$x_i\in X_j$, needs the $i$'th row of matrix $\mathbf{A}^{(t)}$
denoted by $\mathbf{A}_i^{(t)}$ , as
$x_i^{(t)}=\mathbf{A}_i^{(t)}\mathbf{P}^{(t)}$.
 To prevent
unprivileged clients i.e. $\bar{R}_i$ or any other external
eavesdropper to obtain message $x_i$, the vector
$\mathbf{A}_i^{(t)}$ should be delivered privately and securely to
each client $c_u\in R_i$ as a secret key. The process of delivering
these vectors of decoding coefficients $\mathbf{A}_i^{(t)}$ to the
corresponding set of privileged clients $R_i$ is central to this
paper and will be discussed extensively immediately in this section.
\item \emph{Packet Recovery Phase:} As mentioned earlier, we model the channel between the base
station and each client $c_i$ as an erasure channel, i.e. we assume
each encoded element is received by the client $c_i$ with a
probability $1-p_i$. Hence, the missing packets should be either
retransmitted by the base station or provided via cooperation among
the clients by exchanging the missing chunks of information with
each other. Detail of the network coded based retransmission schemes
is not at the scope of this paper (We refer the reader to
\cite{ITW2012} for more information).
\item \emph{Key Sharing} As mentioned earlier, we need to provide
the sets of decoding coefficients privately to privileged clients.
Our method is based on a hybrid private-public key scheme, where an
\emph{initial key} is associated to each message $x_i$. Each initial
key is composed of two components and can be represented as a pair
of functions $\mathcal{K}_i=(\pi_{\tilde{N}}^i, \kappa_i)$. As the
first component of the function $\mathcal{K}_i$, permutation
function is formally defined as follows:
\end{itemize}
\begin{mydef}
A permutation function of the set $\tilde{N}=\{1,\dots,n\}$ is a
one-to-one and covering function denoted by $v=\pi_{\tilde{N}}^i(u)$
which randomly maps each element $u$ in $\tilde{N}$ to an element
$v$ in $\tilde{N}$. $i$ is an arbitrary index which is used later to
identify the index of the corresponding message.
\end{mydef}
\begin{mydef}
 A vector permutation function maps a vector $\mathbf{u}=[u_1 \dots u_n]$ to a
 vector $\mathbf{v}=[v_ 1 \dots v_n]$ such that
 $v_j=\pi_{\tilde{N}}^i(u_j),~\forall j\in \tilde{N}$. Also, we denote
 the vector $[1 \dots n]$ by $\tilde{\mathbf{n}}$.
 \end{mydef}
The second component of the initial key is  $\kappa_i:
\tilde{N}\rightarrow \mathbb{F}_2$ which maps each element $j\in
\tilde{N}$ to $\kappa^i(j),~\forall i,j\in \tilde{N}$ where
$\kappa^i(j)$ has been chosen randomly from a uniform distribution
over the finite field $\mathbb{F}_2$ (i.e.
$prob(\kappa^i(j)=1)=\frac{1}{2}$ and $
prob(\kappa^i(j)=0)=\frac{1}{2}, \forall i,j\in \tilde{N}$. The
initial key $\mathcal{K}_i$ is fixed during all the $T$ rounds of
transmission and is privately provided to a client $c_j$ if $c_j\in
R_i$. The set of initial keys can be either physically delivered to
the clients (e.g. as a part of their hardware), or can be
distributed using a method similar to \cite{SecrecyLaszlo} where it
is expected that after sufficient number of transmissions each set
of privileged clients can pick a key which has not been heard by the
others.


The base station generates a vector of $n$ randomly chosen elements
from $\mathbb{F}_2$ at each round of transmission (or possibly at
the beginning of a period of multiple of transmission rounds)
denoted by $\vec{\nu}^(t)$ called as \emph{regenerating vector} and
broadcast it publicly to all the clients. The decoding coefficients
are determined by this vector for all the messages from the base
station according to the following equation:
\begin{equation}
\label{main}
\mathbf{A}_i^{(t)}=\kappa_i(\pi^i(\tilde{\mathbf{n}}))+\vec{\nu}(t)
\end{equation}
In other words, the initial key $\mathcal{K}_i$ acts as a function
which operates on a publicly announced input, i.e. the regenerating
vector (possibly at each round $t$ or at the beginning of a period
of multiple of transmission rounds) to generate the vector of
decoding coefficients $\mathbf{A}_i^{(t)}$. Therefore the vectors
$\mathbf{A}_i^{(t)}$ can be renewed at each round of transmission
with an overhead of $n$ bits. However, it should be noted that the
same vector $\vec{\nu}^(t)$ is applied to initial keys of all users
$\{\mathcal{K}_1,\dots,\mathcal{K}_n\}$ at round $t$ (otherwise a
huge amount of overhead is imposed).

The outcome of the encoding process over the messages is broadcast
to all the clients and each client $c_i$ would be able to decode its
own message by computing $\mathbf{A}_i^{(t)}\mathbf{P}^{(t)}$. In
the following, some features of the proposed system is briefly
discussed:
\begin{enumerate}
\item We used the linear network codes in a reverse direction, i.e.
instead of generating linear combinations of the messages
($\mathbf{P}=\mathbf{A}\mathbf{X}$) and broadcasting them over the
channel, a system of linear equations
($\mathbf{X}=\mathbf{A}\mathbf{P}$) is solved at the base station to
generate packets $P$. Therefore, each message is related to the set
of packets with a distinct set of decoding coefficients which
enables us to generalize the proposed method to the scenario that
each client is interested in an arbitrary subset of messages without
violating the secrecy of other clients (As each client is provided
with only the decoding coefficients necessary for decoding the
messages that it is privileged for).
\item Reversing the direction of coding scheme mentioned in the last
item, also has the advantage that reduces the complexity at the
receiver side which might potentially have limited power resource
and computational capacity. In our proposed scheme, the receiver
only needs to compute a linear combination of packets for each
message instead of inverting a matrix.
\item The role of regenerating vector is to update the decoding
coefficients to maintain the uniformity of the decoding coefficients
distribution which is necessary for our proof as it will be
discussed in section \ref{sec:Proof}. The role of permutation
functions in the initial keys is to produce a huge space of
possibilities that makes it computationally hard to guess the
decoding coefficients or obtain any information about them. The
amount of leakage of information specially in the case of
non-uniform messages is an interesting topic for further
investigation (for uniform messages some bounds and theorems have
been established in \cite{CryptAnalysisNetCod}). The regenerative
vector updates is expected to play a role in minimizing the leakage
of information in this case. The regenerative vector can be updated
periodically after multiple rounds of transmissions (rather than
updating at each round) but possibly at the price of some
information leakage.
\end{enumerate}

\section{An Example}\label{sec:Example}
In this section, different parts of the proposed system is discussed
through an example. Suppose we have four clients
$C=\{c_1,\dots,c_4\}$ and a set of seven messages
$X=\{x_1,\dots,x_7\}$. Each message is composed of 64 bits,
therefore we have 64 rounds of transmission $t=1,\dots,64$ where at
each round one bit from each message is transmitted to the target
recipients. We assume the sets $\chi_1=\{x_2,x_4,x_7\}$,
$\chi_2=\{x_1,x_3,x_6\}$, $\chi_3=\{x_1,x_2,x_3,x_5,x_6\}$ and
$\chi_4=\{x_2,x_5,x_6,x_7\}$ are demanded by clients $c_1, c_2, c_3$
and $c_4$, respectively.

\begin{table}
\begin{center}
\caption{Initial Keys}
    \begin{tabular}{| l | l | l |}
    \hline
    $\mathcal{K}_1$ &$\pi_{\tilde{N}}^1=[2~ 7~ 6~ 4~ 1~ 5~ 3]$ &$\kappa_1=[1~ 1~ 0~ 1~ 0~ 0~ 1]$\\
    \hline$\mathcal{K}_2$ &$\pi_{\tilde{N}}^2=[4~ 3~ 5~ 1~ 2~ 6~ 7~]$ &$\kappa_2=[1~ 1~ 1~ 0~ 0~ 1~ 0]$\\
    \hline$\mathcal{K}_3$ &$\pi_{\tilde{N}}^3=[5~ 2~ 4~ 7~ 3~ 1~ 6~]$ &$\kappa_3=[0~ 1~ 0~ 0~ 1~ 0~ 0]$\\
    \hline$\mathcal{K}_4$ &$\pi_{\tilde{N}}^4=[3~ 1~ 7~ 6~ 5~ 4~ 2~]$ &$\kappa_4=[1~ 0~ 1~ 0~ 1~ 1~ 1]$\\
    \hline
    $\mathcal{K}_5$ &$\pi_{\tilde{N}}^4=[2~ 1~ 6~ 5~ 7~ 4~ 3~]$ &$\kappa_5=[0~ 1~ 0~ 1~ 0~ 1~ 1]$\\
    \hline
    $\mathcal{K}_6$ &$\pi_{\tilde{N}}^4=[3~ 2~ 5~ 6~ 7~ 2~ 4~]$ &$\kappa_4=[1~ 1~ 0~ 1~ 0~ 1~ 0]$\\
    \hline
    $\mathcal{K}_7$ &$\pi_{\tilde{N}}^4=[1~ 3~ 7~ 5~ 4~ 2~ 6~]$ &$\kappa_4=[1~ 0~ 1~ 0~ 1~ 0~ 1]$\\
    \hline
    \end{tabular}
   \label{tab:MyTable}
\end{center}
\end{table}
Table \ref{tab:MyTable} shows the set of initial key pairs for each
message. Now consider one round of transmission say $t=24$. Suppose
the random regenerating vector produced by the base station for this
round is $\vec{\nu}^{(24)}=[1~0~0~0~1~0~0]$. The base station builds
the matrix $\mathbf{A}^{(t)}$ by computing the equation \ref{main}
for each pair of initial keys $\mathcal{K}_1,\dots, \mathcal{K}_7$
and $\vec{\nu}^{(24)}$ which results in:
\[ \mathbf{A}^{(24)}=\left( \begin{array}{ccccccc}
0 & 1 & 0 & 1 & 0 & 0 &0 \\
1 & 1 & 0 & 1 & 0 & 1 &0 \\
0 & 1 & 0 & 0 & 1 & 0 &0 \\
0 & 1 & 1 & 1 & 0 & 0 &0 \\
0 & 0 & 1 & 0 & 0 & 1 &0 \\
1 & 1 & 0 & 1 & 1 & 1 &1 \\
0 & 1 & 1 & 1 & 1 & 0 &0\\
\end{array} \right)\]
It is easy to check that $\mathrm{det}(\mathbf{A}^{(24)})\neq 0$. If
the base station comes up with a singular matrix, it is deleted and
a new matrix is generated. It should be mentioned that the order of
transmission of the regenerating vector and the coded data elements
($P_i^{(t)}$'s) is not important as long as it is based on a common
protocol between the sender and receivers. Therefore, a batch of
coded elements with their corresponding regenerating vectors might
be packed in a packet (and the result can be encoded using any error
correction codes) and transmitted to all the clients. For instance,
$P_i^{(t)}, \forall t\in \{1,\dots,64\}$ and the $i$'th element of
all regenerating vectors $\vec{\nu}^(t), \forall t=1,\dots,64$ plus
the error correction code bits can be packed as a unit packet $P_i$
by the protocol. Each receiver can acknowledge the reception of each
packet (here by a packet we mean the packed version of the mentioned
components) as it would be extremely costly to send feedback for
each bit and is not practical. However, separate operations are
performed over each bit according to the corresponding decoding
coefficients.

Initial keys can be distributed via any secure private channel or by
using a similar method to \cite{SecrecyLaszlo}. The basis of key
sharing method in \cite{SecrecyLaszlo} is to take the diversity of
packet erasure patterns over the downlink wireless channels between
the base station and the wireless clients as an opportunity to
provide secret keys to the corresponding clients. In
\cite{SecrecyLaszlo}, the base station starts to generate random
messages and broadcast it to all the clients. If a message is
received only by a client $c$ but not by the other ones it can be
used as a key shared between $c$ and the base station. If a key
$\mathcal{K}$ is shared with client $c$, a message $\omega$ is
encoded as $x=\omega\oplus K$ similar to the so-called one-time pad
of Shannon \cite{ShannonSecrecy}. If $x$ is only received by $c$ but
not by the other clients, then $\mathcal{K}$ can be reused,
otherwise $\mathcal{K}$ is burnt and a new key should be used for
sending the next message to $c$. We can use the same approach to
distribute the initial keys, but the initial keys are not required
to be renewed which substantially reduces the the amount of
transmission (and consequently increases the throughput) at the
price of relaxing the secrecy condition to weaker one but practical
yet. To apply the method of \cite{SecrecyLaszlo} to our problem
which is more general in a sense that each client might demand an
arbitrary subset of the messages, the base station should keep
transmitting random messages (of the format
$\mathcal{K}_i=(\pi_{\tilde{N}}^i,\kappa_i)$ until a case is
observed that all clients belonging to $R_i$ have heard it but not
any of the other clients. This might be considered very costly in
terms of throughput efficiency specially if the number of users is
large, however it should be noted that this only happens once at the
beginning and only the regenerating vector $\vec{\nu}^(t)$ is
transmitted publicly at each round afterwards. Therefore if
$T\rightarrow \infty$ the overhead of initial key sharing will tend
to zero. However, as mentioned earlier, the initial keys can be
shared using any type of secret key management method.

As mentioned earlier, the packet recovery can be accomplished either
by the base station or via cooperation. In \cite{ITW2012}, we
assumed that the operations are done over a large field size and
some elements of each row $\mathbf{A}_i$ might have been set to be
zero. Therefore each needs to send a negative acknowledgement (NACK)
for those packets which have not received and need them according to
$\mathbf{A}_i$. As the operations are done over field size 2 in this
paper, if a client sends a NACK only for those packets in its
\emph{wants set} which are not received, some entries in
$\mathbf{A}_i$ would be disclosed which violates the secrecy.
Therefore, each client should send a NACK for all its missing
packets. Then the packets can be recovered using the methods
developed in
\cite{BroadcastNguyen,ARQMedard,OnlineBroadcast,SamehDelay} for
retransmission via the base station or via cooperative data exchange
\cite{ITW2010, ISIT2010,NetCod2011}.

\section{Proof of Secrecy}\label{sec:Proof}
In this section we prove that the aforementioned scheme in section
\ref{sec:System} is weakly secure in an information theoretic sense.
The concept of weakly security introduced in \cite{WeaklySecure}
implies that an unprivileged party can not obtain any
\emph{meaningful} information about a message intended for a group
of privileged users. Weakly security relaxes the perfect secrecy
condition (which does not allow any information to be leaked to an
unauthorized party) to a weaker but more practical condition of
security \cite{WeaklySecure}.

Consider the set of transmitted packets $P^{(t)}$ and also let
$G\subseteq X^{(t)}$.
We assume each client has only received the set of keys which is
privileged for, i.e. $c_i$ or any other external eavesdropper $E$
does not hold $\mathcal{K}_j$ if $c_i\notin R_j$. The following
theorem states the main result of this paper (assuming the
regenerating vector is updated at each round):
\begin{theorem}
\label{MainTheorem}
 An unprivileged client for packet $x_i$ or any
external eavesdropper $E$ can not obtain any information about any
individual message $x_i^{(t)}$, i.e. $I(x_i^{(t)};P^{(t)}|G)=0$,
assuming that $E$ initially holds $G$ and $x_i^{(t)}\notin G$.
\end{theorem}
Before proving the theorem a few lemmas are proven or stated. The
first lemma proved by Gallager is borrowed from
\cite{GallagerLDPC,DistributionFiniteField,TransformBasedDistribution}
where the probability distribution of a linear combination of random
variables over a finite field is studied.
\begin{lemma}
\label{Gallager}
 Let $\beta_1,\dots,\beta_n \in GF(2)$ be random variables
over the field with $prob(\beta_i=1)=\delta_{i1}$, and let
$m_1,\dots,m_n \in GF(2)$. Then the probability distribution of the
linear combination $s=\bigoplus_{i=1}^n m_i \beta_i$ is computed as
follows:
\begin{equation}
\begin{split}
Prob(s=1)=\frac{1-\prod_{i=0}^n (1-2\delta_{i1})}{2},\\
Prob(s=0)=1-Prob(s=1)
\end{split}
\end{equation}
\end{lemma}
\begin{lemma}
\label{Gaussian}
 Suppose that $A^*=[\alpha_{ij}^*]_{n\times n}$ is a matrix
of random elements $\alpha_{ij}^*\in \mathbb{F}_2$, where
$prob(\alpha_{ij}^*=1)=prob(\alpha_{ij}^*=0)=\frac{1}{2},~\forall
i,j\in \tilde{n}$. Let $X^*=A^*P^*$ be a system of linear equations
for known vector $X^*=[x_1^* \dots x_n^*]$ and the vector of
unknowns $P^*=[p_1^* \dots p_n^*]$. Then if the system is rewritten
in the form $\hat{X}=\hat{A}P^*$ using Gaussian elimination method,
where $\hat{A}=[\hat{\alpha}_{ij}]_{n\times n}$ is an upper
triangular matrix, assuming the last entry $p_n^*$ written in the
form $\gamma_{n1}x_1^*+\dots+\gamma_{nn}x_n^*$, then
$prob(\gamma_{ni}=1)=prob(\gamma_{ni}=0)=\frac{1}{2}$.
\end{lemma}
\begin{proof}
To transform the matrix $A^*$ to an upper triangular matrix
$\hat{A}$, row and column operations are applied to $A^*$ in a way
that at the end of the elimination process, all entries
$\alpha_{ij}=0,~\forall j<i-1$.  Depending on the value of the
element $(\ell,i)$, row $j\geq i$ is added to row $\ell$ with
probability $\frac{1}{2}$ (if the element $(\ell,i)$ is $1$,
otherwise no action is required for this element which happens with
probability $\frac{1}{2}$). It should be noted that all the rows
$\hat{\ell}$ that $i<\hat{\ell}<\ell$ might have been affected by
the row $i$ with probability $\frac{1}{2}$ (by affected we mean that
row $i$ has been added to row $\hat{\ell}$). Therefore, in the last
round of elimination process to remove each element $(n,j<n)$, one
of the previous rows might be added to the row $n$ which the row $j$
might be affect even or odd number of times by the row $i<j$ with
equal probabilities (the proof is of equal probabilities is based on
considering all possibilities of being affected by a previous row
and is removed due to space limitations). Even number of being
affected by the $i$'th row updates the coefficient $\gamma_{ni}$ to
be zero and odd number of being affected by the $i$'th row ends up
with $\gamma_{ni}=1$. Therefore
$prob(\gamma_{ni}=1)=prob(\gamma_{ni}=0)=\frac{1}{2}$.
\end{proof}
Now the proof of Theorem \ref{MainTheorem} is established using
Lemma \ref{Gallager} and Lemma \ref{Gaussian}.
\begin{proof}
It is easy to show that
$p[\kappa_i(\pi^i(n_0))=1]+\vec{\nu}(t)=\frac{1}{2}$.
as the regenerating vector is assumed to be drawn from a uniform
distribution. Therefore the elements of $A^{(t)}$ would have a
uniform distribution over $\mathbb{F}_2$. Using lemma
\ref{Gaussian}, it is showed that each packet $P_i^{(t)}$ (which can
be considered as the last element of $A^*$ by swapping the rows) can
be written in the form $\gamma_{i1}x_1^*+\dots+\gamma_{in}x_n^*$ (if
the matrix $A^{(t)}$ is not singular), where
$prob(\gamma_{nj}=1)=prob(\gamma_{nj}=0)=\frac{1}{2}$. Consequently,
using lemma \ref{Gallager}, it is proven that
$prob(P_i^{(t)}=1)=prob(P_i^{(t)}=0)=\frac{1}{2}$. Therefore, each
packet is independent of any individual message $x_j^{(t)}$. Hence,
it can be concluded that $I(x_i^{(t)};P^{(t)}|G)=0$.
\end{proof}


\section*{Acknowledgment}

This work was supported under Australian Research Council Discovery
Projects funding scheme (project no. DP0984950.

\bibliographystyle{IEEEtran}

\end{document}